\documentclass[
a4paper, 
10pt,    
twoside, 
]{LTJournalArticle}

\usepackage{amsmath}
\usepackage{amsfonts}
\usepackage{amssymb}
\usepackage{amstext}
\usepackage{amsthm}
\newtheorem{theorem}{Theorem}

\newtheorem{remark}{Remark}
\newcounter{RomanNumber}

\usepackage{algorithm}
\usepackage{algorithmic}
\usepackage{color}
\usepackage{appendix}
\usepackage{makecell}
\usepackage{multicol}

\title{Deterministic Maximum Likelihood Direction Finding in the Mixture Noise of Gaussian and Spherically Invariant Components}

\author{Mingyan Gong 
}
\date{}

\renewcommand{\maketitlehookd}{
	\begin{abstract}
		\noindent Spherically invariant (SI) random processes can model impulsive noise and unreliable measurements. Recently, the mixture noise of Gaussian and SI components has been used in deterministic maximum likelihood direction finding. In this context, the Expectation-Conditional Maximization (ECM) algorithm, an extension of the expectation-maximization algorithm, has been applied and designed. However, simulation results show that the ECM algorithm always improperly converges. In this article, the ECM Either (ECME) algorithm, an extension of the ECM algorithm, is applied and designed, which additionally utilizes the actual log-likelihood function to first update partial parameter estimates at every iteration and does not need to initialize all parameter estimates. Moreover, the deterministic Cramer–Rao low bounds (CRLBs) of DOA estimators are derived and compared. Simulation results indicate that the ECME algorithm exhibits proper convergence and its root mean square errors of DOA estimates asymptotically approach the CRLBs as the signal powers increase, i.e., the derived CRLBs are correct.

        \par\textbf{Keywords: } Array signal processing, DOA estimation, EM algorithm, impulsive noise, radar signal processing, spherically invariant random process.
	\end{abstract}
}

\begin{document}

\maketitle 

\section{Introduction}

With the development of array signal processing, accurate and robust direction of arrival (DOA) estimation has become more and more crucial for source localization in various applications, e.g., radar, sonar, and wireless communications \cite{Godara}. Obviously, the Gaussian noise assumption is frequently employed in DOA estimation \cite{Krim}, but the Gaussian process may not adequately model some practical noises, e.g., impulsive noise. When conventional methods based on Gaussian noise are applied to DOA estimation in non-Gaussian noise, such mismatch suffers from a dramatic degradation in performance. Therefore, several non-Gaussian distributions have been considered for better modeling impulsive noise, including the symmetric $\alpha$-stable distribution \cite{Shao,Nikias}, Cauchy distribution \cite{Tsakalides2}, generalized Gaussian distribution \cite{Novey}, and complex elliptically symmetric distribution \cite{Ollila,Mecklenbr}. On this foundation, many accurate and robust DOA estimation methods have been proposed, such as FLOM based methods \cite{Tsakalides,THLiu,Swami}, methods based on spatial sign and rank \cite{Visuri}, $\ell_p$-MUSIC \cite{Zeng}, and some recent methods \cite{Cai,Raj}.

Recently, the mixture noise also attracts attention. As the name suggests, the mixture noise is usually the summation of Gaussian noise and a high amplitude component \cite{Middleton}, and here the high amplitude component is used to model impulsive noise or generate outliers. In \cite{DDLee,Lim}, this high amplitude component is also the Gaussian process, thus causing Gaussian mixture noise. In particular, for deterministic maximum likelihood direction finding (DMLDF) in Gaussian mixture noise, the Space-Alternating Generalized Expectation-maximization (SAGE) algorithm, an expectation-maximization (EM) type algorithm \cite{Fessler}, has been applied in \cite{Kozick}. However, the SAGE algorithm may not converge properly when signal powers are unequal, which inspires the alternating expectation/conditional maximization algorithm in \cite{Gong}. In addition to Gaussian mixture noise, the high amplitude component can also be a spherically invariant random process (SIRP) because SIRPs have been successfully developed to model heavy-tailed clutters in adaptive radar detection \cite{KYao,EConte,Barnard,EConte2,Gini,Pascal}. Similarly, for DMLDF in the mixture noise of Gaussian and spherically invariant (SI) components, the Expectation-Conditional Maximization (ECM) algorithm, also an EM-type algorithm \cite{Meng}, has been applied \cite{Yassine}. However, simulation results show that the ECM algorithm always improperly converges.

In this article, the ECM Either (ECME) algorithm \cite{Liu}, an extension of the ECM algorithm, is applied and designed, which additionally utilizes the actual log-likelihood function (LLF) to first update partial parameter estimates at every iteration and does not need to initialize all parameter estimates. Moreover, the deterministic Cramer–Rao low bounds (CRLBs) of DOA estimators are derived and compared. Simulation results indicate that the ECME algorithm exhibits proper convergence and its root mean square errors (RMSEs) of DOA estimates asymptotically approach the CRLBs as the signal powers increase, i.e., the derived CRLBs are correct.

\emph{Notations}: $(\cdot)^T$ and $(\cdot)^H$ mean the transpose and conjugate transpose operators, respectively. $\Vert\textbf{a}\Vert$ stands for the Euclidean norm of vector $\mathbf{a}$. $\mathbf{B}\succ\mathbf{0}_N$ means that the conjugate symmetric matrix $\mathbf{B}$ of order $N$ is positive definite. $\mathrm{Tr}\{\mathbf{B}\}$,  $\mathrm{Det}(\mathbf{B})$, $\mathbf{B}^{1/2}$, and $\mathbf{B}^{-1}$ denote the trace, determinant, square root, and inverse of $\mathbf{B}$, respectively. $\mathbf{B}_{i,j}$ is the element in the $i$th row and $j$th column
of $\mathbf{B}$. $\mathcal{E}\{\cdot\}$ and $\mathcal{D}\{\cdot\}$ denote the expectation and covariance operators, respectively. $\Re\{\mathbf{B}\}$ and $\Im\{\mathbf{B}\}$ mean the real and imaginary parts of $\mathbf{B}$, respectively. $\hat{\beta}$ denotes the estimate of unknown parameter $\beta$, $\bar{\beta}$ is usually $\hat{\beta}$ updated at the previous iteration, and $\breve{\beta}$ is $\hat{\beta}$ updated at the current iteration. Finally, $\mathbf{0}$ is an all-zero vector or matrix, $\jmath^2=-1$, and $\mathbf{I}_N$ is the identity matrix of order $N$.

\section{DMLDF Problem}

We consider a uniform linear array (ULA) consisting of $L$ omnidirectional antenna elements, which receives $G~(G<L)$ narrow-band source signals impinging from far-field at distinct directions $\beta_g$'s, where $\beta_g$ stands for the DOA of the $g$th source signal. Then, we let the spacing between adjacent antennas be the half-wavelength of these source signals so that the signal received at this ULA can be modeled as
\begin{eqnarray}  \label{1}  
\mathbf{r}(t)=\sum_{g=1}^G\mathbf{d}(\beta_g)u_g(t)+\mathbf{n}(t)=\mathbf{D}(\boldsymbol{\beta})\mathbf{u}(t)+\mathbf{n}(t),
\end{eqnarray}
where $\mathbf{d}(\beta_g)=[1~e^{-\jmath\pi\cos(\beta_g)}~\cdots~e^{-\jmath(L-1)\pi\cos(\beta_g)}]^T$ stands for the direction vector of the $g$th source signal, and $\mathbf{D}(\boldsymbol{\beta})=[\mathbf{d}(\beta_1)~\cdots~\mathbf{d}(\beta_G)]$ with $\boldsymbol{\beta}=[\beta_1~\cdots~\beta_G]^T$. Additionally, $\mathbf{u}(t)=[u_1(t)~\cdots~u_G(t)]^T$ represents the source signal vector.

\subsection{Mixture Noise of Gaussian and SI Components}

In \eqref{1}, $\mathbf{n}(t)$ denotes the mixture noise vector of Gaussian and SI components and is constructed as
\begin{eqnarray}  \label{2}  
&\mathbf{n}(t)=\mathbf{w}(t)+\mathbf{v}(t)=\mathbf{w}(t)+\sqrt{\tau(t)}\mathbf{o}(t),
\end{eqnarray}
where $\mathbf{w}(t)\sim\mathcal{CN}(\mathbf{0},\sigma^2\mathbf{I}_L)$ denotes the Gaussian component while $\mathbf{v}(t)$ means the SI component. Here, $\mathbf{o}(t)\sim\mathcal{CN}(\mathbf{0},\boldsymbol{\Sigma})$ and the texture process $\tau(t)>0$ is assumed to be deterministic and unknown for simplicity \cite{EConte2,Besson}, which leads to $\mathbf{v}(t)\sim\mathcal{CN}\big(\mathbf{0},\tau(t)\boldsymbol{\Sigma}\big)$. Importantly, we let $\mathbf{w}(t)$ and $\mathbf{o}(t)$ be mutually uncorrelated and have
\begin{eqnarray}  \label{3}  
    &\mathbf{n}(t)\sim\mathcal{CN}(\mathbf{0},\mathbf{Q}_t),\mathbf{Q}_t=\sigma^2\mathbf{I}_L+\tau(t)\boldsymbol{\Sigma}\succ\mathbf{0}_L.
\end{eqnarray}

We assume that $\mathbf{u}(t)$ in \eqref{1} is also deterministic and unknown, i.e., the well-known deterministic source signal model \cite{Ziskind}. Hence, $\mathbf{r}(t)\sim\mathcal{CN}\big(\mathbf{D}(\boldsymbol{\beta})\mathbf{u}(t),\mathbf{Q}_t\big)$ and the probability density function of $\mathbf{r}(t)$ can be written by
\begin{eqnarray}  \label{4}  
&p\big(\mathbf{r}(t);\boldsymbol{\Phi}_t\big)=
\frac{1}{\pi^L\mathrm{Det}(\mathbf{Q}_t)}\times\nonumber\\
&\exp\Big(-\big[\mathbf{r}(t)-\mathbf{D}\mathbf{u}(t)\big]^H\mathbf{Q}^{-1}_t\big[\mathbf{r}(t)-\mathbf{D}\mathbf{u}(t)\big]\Big),
\end{eqnarray} 
where $\boldsymbol{\Phi}_t=\big(\boldsymbol{\beta},\mathbf{u}(t),\sigma,\tau(t),\boldsymbol{\Sigma}\big)$ and $\mathbf{D}(\boldsymbol{\beta})$ is replaced with $\mathbf{D}$ for notational brevity.

\subsection{Problem Exposition}

In practice, we have to sample $\mathbf{r}(t)$ and collect statistically independent measurements (known as incomplete data in EM-type algorithms \cite{GMcLachlan}):
\begin{eqnarray}   \label{5}  
&\mathbf{r}(t)=\mathbf{D}\mathbf{u}(t)+\mathbf{n}(t),t=1,\dots,T,
\end{eqnarray}
where $T$ represents the number of measurements. Obviously, estimating $\boldsymbol{\beta}$ requires manipulating the measurements $\mathbf{R}$ with $\mathbf{R}=[\mathbf{r}(1)~\cdots~\mathbf{r}(T)]$. To proceed, we employ the method of maximum likelihood (ML) estimation \cite{Kay} and write the LLF of $\mathbf{R}$ by
\begin{eqnarray}   \label{6}  
&f(\boldsymbol{\Phi})=\ln p(\mathbf{R};\boldsymbol{\Phi})={\sum}_{t=1}^T\ln p\big(\mathbf{r}(t);\boldsymbol{\Phi}_t\big)\nonumber\\
&=-TL\ln(\pi)-{\sum}_{t=1}^T\ln\mathrm{Det}(\mathbf{Q}_t)\nonumber\\
&-{\sum}_{t=1}^T\big[\mathbf{r}(t)-\mathbf{D}\mathbf{u}(t)\big]^H\mathbf{Q}^{-1}_t\big[\mathbf{r}(t)-\mathbf{D}\mathbf{u}(t)\big],
\end{eqnarray}
where $\boldsymbol{\Phi}=(\boldsymbol{\beta},\mathbf{U},\sigma,\boldsymbol{\tau},\boldsymbol{\Sigma})$ with $\mathbf{U}=[\mathbf{u}(1)~\cdots~\mathbf{u}(T)]$ and $\boldsymbol{\tau}=[\tau(1)~\cdots~\tau(T)]^T$. Then, this method requires us to solve the non-convex DMLDF problem:
\begin{eqnarray}  \label{7}  
&\max_{\boldsymbol{\beta}\in(0,\pi)^G,\mathbf{U},\sigma>0,\boldsymbol{\tau}>\mathbf{0},\boldsymbol{\Sigma}\succ\mathbf{0}_L}f(\boldsymbol{\Phi}).
\end{eqnarray}

\section{ECM Algorithm}

The ECM algorithm has been applied and designed to solve problem \eqref{7} efficiently \cite{Yassine}, but simulation results show that it always improperly converges. In this section, we detail this ECM algorithm.

\subsection{Complete Data}

To begin with, since the ECM algorithm needs to determine its complete data, which contains missing data, we utilize the $\mathbf{v}(t)$'s as the missing data and construct the complete data as
\begin{eqnarray}   \label{8} 
&\mathbf{X}=(\mathbf{R},\mathbf{V}),\mathbf{V}=[\mathbf{v}(1)~\cdots~\mathbf{v}(T)].
\end{eqnarray}
Accordingly, the complete-data LLF is
\begin{eqnarray}  \label{9}  
&&Q(\mathbf{X};\boldsymbol{\Phi})=\ln p(\mathbf{X};\boldsymbol{\Phi})={\sum}_{t=1}^T\ln p\big(\mathbf{r}(t),\mathbf{v}(t);\boldsymbol{\Phi}_t\big)\nonumber\\
&&=\sum_{t=1}^T\ln\Big[p\big(\mathbf{r}(t)\mid\mathbf{v}(t);\boldsymbol{\beta},\mathbf{u}(t),\sigma\big)p\big(\mathbf{v}(t);\tau(t),\boldsymbol{\Sigma}\big)\Big]
\nonumber\\
&&=-2TL\ln(\pi)-TL\ln(\sigma^2)-T\ln\mathrm{Det}(\boldsymbol{\Sigma})\nonumber\\
&&-\frac{1}{\sigma^2}{\sum}_{t=1}^T\big\Vert\mathbf{r}(t)-\mathbf{Du}(t)-\mathbf{v}(t)\big\Vert^2
\nonumber\\
&&-{\sum}_{t=1}^T\Big[L\ln\big(\tau(t)\big)+\frac{1}{\tau(t)}\mathbf{v}^H(t)\boldsymbol{\Sigma}^{-1}\mathbf{v}(t)\Big].
\end{eqnarray}

\subsection{Procedure}

Without loss of generality, we describe one iteration of the ECM algorithm, which consists of one expectation step (E-step) and three conditional maximization steps (CM-steps).

\subsubsection{E-step}

Given $\bar{\boldsymbol{\Phi}}=(\bar{\boldsymbol{\beta}},\bar{\mathbf{U}},\bar{\sigma},\bar{\boldsymbol{\tau}},\bar{\boldsymbol{\Sigma}})$, this step computes the conditional expectation of $Q(\mathbf{X};\boldsymbol{\Phi})$ in \eqref{9} by
\begin{eqnarray}  \label{10}  
&Q(\boldsymbol{\Phi};\bar{\boldsymbol{\Phi}})=\mathcal{E}\big\{Q(\mathbf{X};\boldsymbol{\Phi})\mid\mathbf{R};\bar{\boldsymbol{\Phi}}\big\}\nonumber\\
&=-2TL\ln(\pi)-TL\ln(\sigma^2)-T\ln\mathrm{Det}(\boldsymbol{\Sigma})\nonumber\\
&-\frac{1}{\sigma^2}\sum_{t=1}^T\Big[\big\Vert\mathbf{r}(t)-\mathbf{Du}(t)-\bar{\boldsymbol{\mu}}(t)\big\Vert^2+\mathrm{Tr}\big\{\bar{\boldsymbol{\Pi}}_t\big\}\Big]\nonumber\\
&-\sum_{t=1}^T\Big[L\ln\big(\tau(t)\big)+\frac{1}{\tau(t)}\mathrm{Tr}\big\{\boldsymbol{\Sigma}^{-1}\bar{\boldsymbol{\Delta}}_t\big\}\Big],
\end{eqnarray}
where $\bar{\mathbf{Q}}_t=\bar{\sigma}^2\mathbf{I}_L+\bar{\tau}(t)\bar{\boldsymbol{\Sigma}}$, $\bar{\boldsymbol{\Delta}}_t=\bar{\boldsymbol{\mu}}(t)\bar{\boldsymbol{\mu}}^H(t)+\bar{\boldsymbol{\Pi}}_t$, $\bar{\boldsymbol{\mu}}(t)$ and $\bar{\boldsymbol{\Pi}}_t$ denote the conditional mean and covariance matrix of $\mathbf{v}(t)$, respectively, i.e.,
\begin{subequations}
\begin{eqnarray}  \label{11a}  
    \bar{\boldsymbol{\mu}}(t)=\mathcal{E}\big\{\mathbf{v}(t)\mid\mathbf{R};\bar{\boldsymbol{\Phi}}\big\}=\bar{\tau}(t)\bar{\boldsymbol{\Sigma}}\bar{\mathbf{Q}}_t^{-1}\big[\mathbf{r}(t)-\bar{\mathbf{D}}\bar{\mathbf{u}}(t)\big],\nonumber\\
\bar{\boldsymbol{\Pi}}_t=\mathcal{D}\big\{\mathbf{v}(t)\mid\mathbf{R};\bar{\boldsymbol{\Phi}}\big\}=\bar{\tau}(t)\bar{\boldsymbol{\Sigma}}-\bar{\tau}^2(t)\bar{\boldsymbol{\Sigma}}\bar{\mathbf{Q}}_t^{-1}\bar{\boldsymbol{\Sigma}}.\nonumber 
\end{eqnarray}  
\end{subequations}

\subsubsection{First CM-step}

This step only updates $\hat{\boldsymbol{\tau}}$ by maximizing $Q(\boldsymbol{\Phi};\bar{\boldsymbol{\Phi}})$ in \eqref{10} with respect to $\boldsymbol{\tau}$ while keeping $(\boldsymbol{\beta},\mathbf{U},\sigma,\boldsymbol{\Sigma})=(\bar{\boldsymbol{\beta}},\bar{\mathbf{U}},\bar{\sigma},\bar{\boldsymbol{\Sigma}})$, which leads to
\begin{eqnarray}  \label{12}  
&\min_{\tau(t)>0}L\ln\big(\tau(t)\big)+\frac{1}{\tau(t)}\mathrm{Tr}\big\{\bar{\boldsymbol{\Sigma}}^{-1}\bar{\boldsymbol{\Delta}}_t\big\},
\end{eqnarray}
where $t=1,\dots,T$. The solutions are
\begin{eqnarray}  \label{13}  
&\breve{\tau}(t)=\mathrm{Tr}\big\{\bar{\boldsymbol{\Sigma}}^{-1}\bar{\boldsymbol{\Delta}}_t\big\}/L>0,t=1,\dots,T.
\end{eqnarray}

\subsubsection{Second CM-step}
This step only updates $\hat{\boldsymbol{\Sigma}}$ by maximizing $Q(\boldsymbol{\Phi};\bar{\boldsymbol{\Phi}})$ in \eqref{10} with respect to $\boldsymbol{\Sigma}$ while keeping $(\boldsymbol{\beta},\mathbf{U},\sigma,\boldsymbol{\tau})=(\bar{\boldsymbol{\beta}},\bar{\mathbf{U}},\bar{\sigma},\breve{\boldsymbol{\tau}})$, which leads to
\begin{eqnarray}  \label{14} 
&\min_{\boldsymbol{\Sigma}\succ\mathbf{0}_L}\ln\mathrm{Det}(\boldsymbol{\Sigma})+\mathrm{Tr}\big\{\boldsymbol{\Sigma}^{-1}\cdot\frac{1}{T}\sum_{t=1}^T\frac{\bar{\boldsymbol{\Delta}}_t}{\breve{\tau}(t)}\big\}.
\end{eqnarray}
The solution is\footnote{For simple optimization, we do not impose the normalized constraint for $\boldsymbol{\Sigma}$, i.e., $\mathrm{Tr}\{\boldsymbol{\Sigma}\}=L$.}
\begin{eqnarray}   \label{15}  
&\breve{\boldsymbol{\Sigma}}=\frac{1}{T}\sum_{t=1}^T\frac{\bar{\boldsymbol{\Delta}}_t}{\breve{\tau}(t)}\succ\mathbf{0}_L. 
\end{eqnarray}

\subsubsection{Third CM-step}
This step only updates $(\hat{\boldsymbol{\beta}},\hat{\mathbf{U}},\hat{\sigma})$ by maximizing $Q(\boldsymbol{\Phi};\bar{\boldsymbol{\Phi}})$ in \eqref{10} with respect to $(\boldsymbol{\beta},\mathbf{U},\sigma)$ while keeping $(\boldsymbol{\tau},\boldsymbol{\Sigma})=(\breve{\boldsymbol{\tau}},\breve{\boldsymbol{\Sigma}})$, which leads to
\begin{eqnarray}  \label{16}  
&\min_{\boldsymbol{\beta}\in(0,\pi)^G,\mathbf{U},\sigma>0}TL\ln(\sigma^2)+\nonumber\\
&\frac{1}{\sigma^2}{\sum}_{t=1}^T\Big[\big\Vert\bar{\mathbf{e}}(t)-\mathbf{D}\mathbf{u}(t)\big\Vert^2+\mathrm{Tr}\big\{\bar{\boldsymbol{\Pi}}_t\big\}\Big].
\end{eqnarray}
where $\bar{\mathbf{e}}(t)=\mathbf{r}(t)-\bar{\boldsymbol{\mu}}(t)$. This can be solved by a separable fashion and the solution is
\begin{subequations}   \label{17}
\begin{eqnarray}   \label{17a}  
&\breve{\boldsymbol{\beta}}=\arg\max_{\boldsymbol{\beta}\in(0,\pi)^G}\mathrm{Tr}\big\{\mathbf{P}_{\mathbf{D}}\hat{\mathbf{R}}_{\bar{\mathbf{e}}}\big\},\\
&\breve{\mathbf{u}}(t)=(\breve{\mathbf{D}}^H\breve{\mathbf{D}})^{-1}\breve{\mathbf{D}}^H\bar{\mathbf{e}}(t),t=1,\dots,T,\\
&\breve{\sigma}^2=\alpha_1\Big(\frac{1}{L}\mathrm{Tr}\big\{\mathbf{P}_{\breve{\mathbf{D}}}^{\perp}\hat{\mathbf{R}}_{\bar{\mathbf{e}}}\big\}+\frac{1}{TL}\sum_{t=1}^T\mathrm{Tr}\big\{\bar{\boldsymbol{\Pi}}_t\big\}\Big)\nonumber\\
&+(1-\alpha_1)\bar{\sigma}^2>0,
\end{eqnarray}
\end{subequations}
where $\hat{\mathbf{R}}_{\bar{\mathbf{e}}}=(1/T)\sum_{t=1}^T\bar{\mathbf{e}}(t)\bar{\mathbf{e}}^H(t)$, $\mathbf{P}_{\mathbf{D}}=\mathbf{D}(\mathbf{D}^H\mathbf{D})^{-1}\mathbf{D}^H$ is the orthogonal projector onto the column-space of $\mathbf{D}$, and $\mathbf{P}_{\breve{\mathbf{D}}}^{\perp}=\mathbf{I}_L-\mathbf{P}_{\breve{\mathbf{D}}}$. Importantly, $\alpha_1\in[0,1]$ is used to adjust the difference between $\breve{\sigma}^2$ and $\bar{\sigma}^2$ and avoid extremely small $\breve{\sigma}^2$. 


\begin{remark}
Simulation results show that although the initialization or initial point is good, the ECM algorithm cannot properly converge, as shown in \textbf{Figure \ref{f0}} and \textbf{Figure \ref{f1}}. Moreover, the ECM algorithm only utilizes the complete-data LLF $Q(\mathbf{X};\boldsymbol{\Phi})$ in \eqref{9} to update parameter estimates at every iteration, which yields slow convergence \cite{Liu, GMcLachlan}.
\end{remark}

\section{ECME Algorithm}

To solve the issue of improper convergence from the ECM algorithm, we apply and design the ECME algorithm in this section.

\subsection{Procedure}

We design the ECME algorithm, whose every iteration is composed of one E-step and four CM-steps. Specifically,
\begin{itemize}
    \item the first CM-step directly utilizes the actual LLF $f(\boldsymbol{\Phi})$ in \eqref{6} to update $\hat{\mathbf{U}}$ and $\hat{\sigma}$ by two explicit formulas,
    \item the other steps are the same steps at every iteration of the ECM algorithm
\end{itemize}
Thus, we only describe the first CM-step below.

In order to perform the first CM-step conveniently, we adopt a new parametrization of $\mathbf{Q}_t$ in \eqref{3} by
\begin{eqnarray}  \label{18}  
&\mathbf{Q}_t=\sigma^2\mathbf{C}_t=\sigma^2\big(\mathbf{I}_L+\zeta(t)\boldsymbol{\Sigma}\big),\zeta(t)=\tau(t)/\sigma^2,
\end{eqnarray}
and have the new parameter $\boldsymbol{\zeta}=[\zeta(1)~\cdots~\zeta(T)]^T>\mathbf{0}$.

\subsubsection{First CM-step} \label{4.1.1}

Given $(\bar{\boldsymbol{\beta}},\bar{\boldsymbol{\zeta}},\bar{\boldsymbol{\Sigma}})$, this step only updates $(\hat{\mathbf{U}},\hat{\sigma})$ by directly maximizing $f(\boldsymbol{\Phi})$ in \eqref{6} with respect to $(\mathbf{U},\sigma)$, which leads to
\begin{eqnarray}   \label{19}  
&\min_{\mathbf{U},\sigma>0}TL\ln(\sigma^2)+\nonumber\\
&\frac{1}{\sigma^2}\sum_{t=1}^T\big[\mathbf{r}(t)-\bar{\mathbf{D}}\mathbf{u}(t)\big]^H\bar{\mathbf{C}}^{-1}_t\big[\mathbf{r}(t)-\bar{\mathbf{D}}\mathbf{u}(t)\big],
\end{eqnarray}
where $\bar{\mathbf{C}}_t=\mathbf{I}_L+\bar{\zeta}(t)\bar{\boldsymbol{\Sigma}}$. This problem can be solved in a separable fashion and the solution is
\begin{subequations}   \label{20}  
\begin{eqnarray}   \label{20a}
\bar{\mathbf{u}}(t)=(\bar{\mathbf{D}}^H\bar{\mathbf{C}}_t^{-1}\bar{\mathbf{D}})^{-1}\bar{\mathbf{D}}^H\bar{\mathbf{C}}_t^{-1}\mathbf{r}(t), t=1,\dots,T,\\
\bar{\sigma}^2=(1-\alpha_2)\bar{\sigma}^2+\alpha_2\times\nonumber\\
\Big(\frac{1}{TL}\sum_{t=1}^T\big[\mathbf{r}(t)-\bar{\mathbf{D}}\bar{\mathbf{u}}(t)]^H\bar{\mathbf{C}}_t^{-1}\big[\mathbf{r}(t)-\bar{\mathbf{D}}\bar{\mathbf{u}}(t)]\Big),
\end{eqnarray}
\end{subequations}
where $\alpha_2$ has the same purpose as $\alpha_1$. After obtaining $(\bar{\mathbf{U}},\bar{\sigma})$, the first CM-step ends and we have $\bar{\boldsymbol{\Phi}}=(\bar{\boldsymbol{\beta}},\bar{\mathbf{U}},\bar{\sigma},\bar{\boldsymbol{\tau}},\bar{\boldsymbol{\Sigma}})$ with $\bar{\boldsymbol{\tau}}=\bar{\sigma}^2\bar{\boldsymbol{\zeta}}$.


\subsection{Convergence, Initialization, and Complexity}

According to \cite{GMcLachlan,Liu,Wu}, we know that the ECME algorithm satisfies certain ``regularity'' conditions and its sequence of parameter estimates always converges to a stationary point of $f(\boldsymbol{\Phi})$. Moreover, the ECME algorithm additionally uses the actual LLF $f(\boldsymbol{\Psi})$ to update $(\mathbf{U},\sigma)$ at every iteration and thus yields faster convergence than the ECM algorithm.

Eq. \eqref{19} indicates that the ECME algorithm only needs to initialize $(\hat{\boldsymbol{\beta}},\hat{\sigma}, \hat{\boldsymbol{\zeta}},\hat{\boldsymbol{\Sigma}})$ without $\hat{\mathbf{U}}$, so it initializes fewer parameter estimates than the ECM algorithm. Moreover, the computational complexities of each iteration in both algorithms are dominated by updating $\hat{\boldsymbol{\beta}}$ in \eqref{17a}, which usually requires a $G$-dimensional search. Hence, the ECME algorithm nearly possesses the same computational complexity of every iteration as the ECM algorithm and is computationally more efficient due to the faster convergence.

\section{CRLB derivation}

In this section, we derive the deterministic CRLB of $\hat{\boldsymbol{\beta}}$ in the mixture noise of Gaussian and SI components and compare it with the well-known CRLB derived by Pesavento and Gershman \cite{Pesavento}.

The following \textbf{Theorem \ref{t1}} gives a closed-form expression for the deterministic CRBL of $\hat{\boldsymbol{\beta}}$ in the mixture noise of Gaussian and SI components.
\begin{theorem}  \label{t1}
    In the mixture noise of Gaussian and SI components, the deterministic CRLB of $\hat{\boldsymbol{\beta}}$ is given by
    \begin{eqnarray}  \label{21}
        &\mathrm{CRLB}(\hat{\boldsymbol{\beta}})=\Big(2\sum_{t=1}^T\Re\big\{\mathbf{U}^H(t)\tilde{\mathbf{Z}}_t^H\mathbf{P}_{\tilde{\mathbf{D}}_t}^{\perp}\tilde{\mathbf{Z}}_t\mathbf{U}(t)\big\}\Big)^{-1}
    \end{eqnarray}
    where $\tilde{\mathbf{Z}}_t=\mathbf{Q}_t^{-1/2}\mathbf{Z}$, $ \mathbf{Z}=\big[\frac{d}{d\beta_1}\mathbf{d}(\beta_1)~\cdots~\frac{d}{d\beta_G}\mathbf{d}(\beta_G)\big]$, $\tilde{\mathbf{D}}_t=\mathbf{Q}_t^{-1/2}\mathbf{D}$,$\mathbf{U}(t)=\mathrm{diag}\{u_1(t),\dots,u_G(t)\}$, $ \mathbf{P}_{\tilde{\mathbf{D}}_t}^{\perp}=\mathbf{I}_L-\tilde{\mathbf{D}}_t(\tilde{\mathbf{D}}_t^H\tilde{\mathbf{D}}_t)^{-1}\tilde{\mathbf{D}}_t^H$. Thus, the CRLB of $\hat{\beta}_g$ is $\big[\mathrm{CRLB}(\hat{\boldsymbol{\beta}})\big]_{g,g}$ and
    \begin{eqnarray} \label{22}
        \sqrt{\mathrm{D}\{\hat{\beta}_g\}}\ge\sqrt{\big[\mathrm{CRLB}(\hat{\boldsymbol{\beta}})\big]_{g,g}},g=1,\dots,G.
    \end{eqnarray}
\end{theorem}

\begin{proof}
   See Appendix \ref{fulu}.
\end{proof}

Based on \eqref{21}, we compare it with the deterministic CRLB of $\hat{\boldsymbol{\beta}}$ in unknown nonuniform noise given by \cite{Pesavento}
\begin{eqnarray}  \label{23}
        &\mathrm{CRLB}(\hat{\boldsymbol{\beta}})=\Big(2\sum_{t=1}^T\Re\big\{\mathbf{U}^H(t)\ddot{\mathbf{Z}}^H\mathbf{P}_{\ddot{\mathbf{D}}}^{\perp}\ddot{\mathbf{Z}}\mathbf{U}(t)\big\}\Big)^{-1}
\end{eqnarray}
where $\ddot{\mathbf{Z}}=\mathbf{Q}^{-1/2}\mathbf{Z}$, $\ddot{\mathbf{D}}=\mathbf{Q}^{-1/2}\mathbf{D}$, and $\mathbf{Q}$ means the noise covariance matrix. Comparing \eqref{21} and \eqref{23}, we can observe that both are similar and, respectively, involve the transformed array manifolds $\tilde{\mathbf{D}}_t$ and $\ddot{\mathbf{D}}$ for prewhitening. Importantly, the main difference is that the noise covariance matrix $\mathbf{Q}_t$ in the mixture noise has more unknown parameters and is time-dependent due to the texture process $\tau(t)$, which increases the number of unknown parameters and makes DOA estimation more complicated.

\section{Simulation Results}

In this section, simulation results are used to show the ECME algorithm and verify the CRLBs in \eqref{22}. Unless otherwise specified in the following figures, we adopt $L=8$, $G=2$, $\beta_1=40^{\circ}$, $\beta_2=70^{\circ}$, $T=50$, $u_1(t)=5$, $u_2(t)=10$. The noise parameters are $\sigma=1$, $\boldsymbol{\Sigma}_{i,j}=10\times0.8^{\vert i-j\vert}e^{\jmath0.2\pi(i-j)}$, and $\tau(t)$ follows the exponential distribution with mean $0.1$. Moreover, initialize $\hat{\beta}_1=35^{\circ}$, $\hat{\beta}_2=65^{\circ}$, $\hat{u}_1(t)=\hat{u}_2(t)=1$, $\hat{\sigma}=1$, $\hat{\boldsymbol{\Sigma}}=\mathbf{I}_L$, $\hat{\tau}(t)=1$, $\hat{\zeta}(t)=1$, $\alpha_1=\alpha_2=0.1$, and the number of iterations for each algorithm is $60$.

\subsection{Search Method for Problem \eqref{17a}}

Problem \eqref{17a} is classical and can be solved by the well-known alternating projection (AP) algorithm \cite{Ziskind}. Therefore, we adopt the initialization method and the AP algorithm for problems \eqref{17a}.

We give the initialization method in \textbf{Algorithm \ref{alg:C}} and the AP algorithm in \textbf{Algorithm \ref{alg:D}} for problem \eqref{17a}. Moreover, we design \textbf{Algorithm \ref{alg:E}} to solve
\begin{eqnarray}  \label{24}   
&\breve{\beta}_g=\arg\max_{\beta_g\in(0,\pi)}\mathrm{Tr}\{\mathbf{P}_{\mathbf{A}}\hat{\mathbf{R}}_{\bar{\mathbf{e}}}\}\nonumber\\
&=\arg\max_{\beta_g\in(0,\pi)}f_g\big(-\pi\cos(\beta_g)\big),
\end{eqnarray}
where $\mathbf{A}_g$ consists of the column vector(s) in $\mathbf{A}$ except $\mathbf{d}(\beta_g)$ or $\mathbf{d}(\hat{\beta}_g)$, $\mathbf{P}^{\perp}_{\mathbf{A}_g}=\mathbf{I}_L-\mathbf{P}_{\mathbf{A}_g}$, and
\begin{eqnarray}  \label{25}  
    f_g\big(-\pi\cos(\beta_g)\big)=\frac{\mathbf{b}^H(\beta_g)\mathbf{P}^{\perp}_{\mathbf{A}_g}\hat{\mathbf{R}}_{\bar{\mathbf{e}}}\mathbf{P}^{\perp}_{\mathbf{A}_g}\mathbf{b}(\beta_g)}{\mathbf{b}^H(\beta_g)\mathbf{P}^{\perp}_{\mathbf{A}_g}\mathbf{b}(\beta_g)}.
\end{eqnarray}

\begin{algorithm}
\caption{Initialization Method} \label{alg:C}
\begin{algorithmic}[1]
\STATE{Initialize $\mathbf{A}=[\mathbf{b}(\theta_1)]$ and $g=1$.}
\WHILE{$g\le G$}
    \STATE{Obtain $\hat{\theta}_g$ using \textbf{Algorithm \ref{alg:E}}.}
    \STATE{$g=g+1$ and $\mathbf{A}=[\mathbf{A}~\mathbf{b}(\theta_{g})]$.}
\ENDWHILE
\STATE{ Output $\hat{\boldsymbol{\beta}}=[\hat{\beta}_1~\cdots~\hat{\beta}_G]^T$. }
\end{algorithmic}
\end{algorithm}

\begin{algorithm}
\caption{AP Algorithm} \label{alg:D}
\begin{algorithmic}[1]
\STATE{Initialize $\hat{\boldsymbol{\beta}}$ using \textbf{Algorithm \ref{alg:C}}, $c=1$.}
\STATE{Let $\mathbf{A}=[\mathbf{b}(\hat{\beta}_1)~\cdots~\mathbf{b}(\hat{\beta}_G)]$.}
\WHILE{$c\le 40$}
    \FOR{$g=1,\dots,G$}
        \STATE{Update $\hat{\beta}_g$ using \textbf{Algorithm \ref{alg:E}}.}
    \ENDFOR
    \STATE {$c=c+1$.}
\ENDWHILE
\STATE{ Output $\hat{\boldsymbol{\beta}}=[\hat{\beta}_1~\cdots~\hat{\beta}_G]^T$.}
\end{algorithmic}
\end{algorithm}

\begin{algorithm}
\caption{Golden Section Search for Problem \eqref{25}} \label{alg:E}
\begin{algorithmic}[1]
\STATE{Initialize $\Delta u=2\pi/10^3$, $u_g=-\pi+\Delta u$, and $f_{\mathrm{max}}=-\infty$.}
\WHILE{$u_g<\pi$}
    \IF{$f_g(u_g)>f_{\mathrm{max}}$}
        \STATE{$u_g^*=u_g$ and $f_{\mathrm{max}}=f_g(u_g)$.}
    \ENDIF
    \STATE{$u_g=u_g+\Delta u$.}
\ENDWHILE
\STATE{$start=u_g^*-\Delta u$ and $end=u_g^*+\Delta u$.}
\WHILE{$end-start>10^{-4}$}
\STATE{$mid_1=start+0.382(end-start)$ and $mid_2=start+0.618(end-start)$.}
\IF{$f_g(mid_1)<f_g(mid_2)$}
   \STATE{$start=mid_1$.}
\ELSE
   \STATE{$end=mid_2$.}
\ENDIF
\ENDWHILE
\STATE{$\breve{u}_g=(start+end)/2$ and $\breve{\beta}_g=\cos^{-1}\big(\breve{u}_g/(-\pi)\big)$.}
\end{algorithmic}
\end{algorithm}

\subsection{Algorithm Evaluation and CRLB Verification}

Without loss of generality, \textbf{Figure \ref{f0}} compares the convergence of both algorithms given the same $\mathbf{R}$. We can observe that the ECM algorithm improperly converges although the initial point is good. On the contrary, the ECME algorithm properly converges. \textbf{Figure \ref{f1}} shows a scatter plot of $(\hat{\beta}_1,\hat{\beta}_2)$'s from both algorithms under $50$ trials, and we can see that the $(\hat{\beta}_1,\hat{\beta}_2)$'s from the ECME algorithm center on the true value $(40^{\circ},70^{\circ})$ more closely. Hence, we analyze the accuracy of the ECME algorithm in \textbf{Figure \ref{f2}}.

\begin{figure} 
\centerline{\includegraphics[width=19pc]{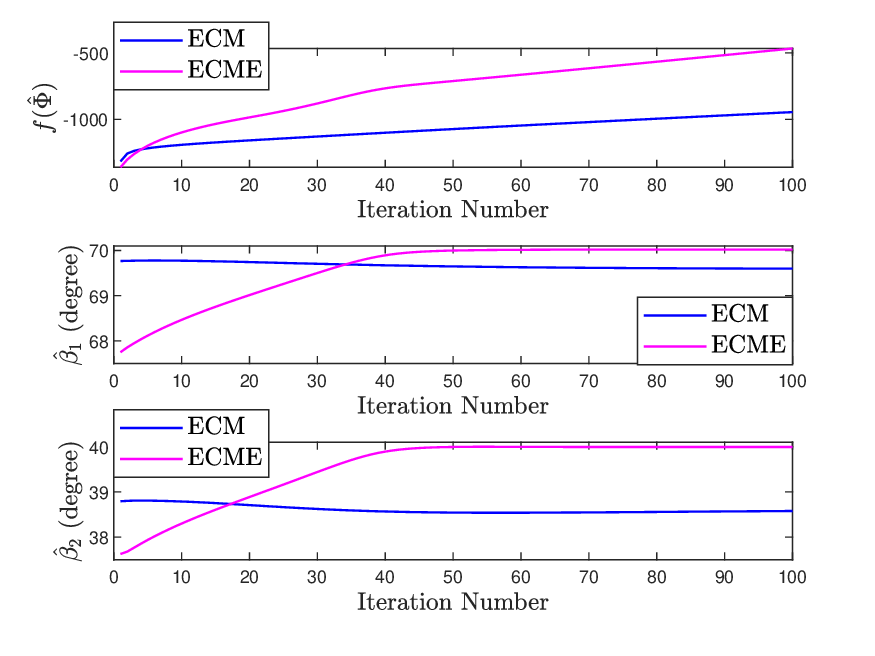}}
\caption{Convergence comparison of both algorithms.}\label{f0}
\end{figure}

\begin{figure} 
\centerline{\includegraphics[width=19pc]{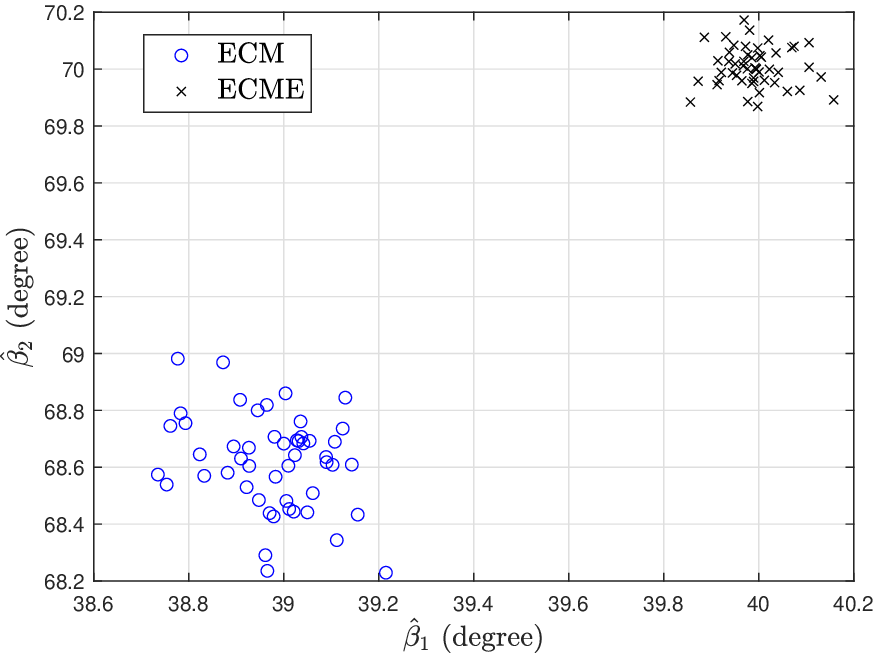}}
\caption{Scatter plot of $(\hat{\beta}_1,\hat{\beta}_2)$'s obtained from both algorithms.}\label{f1}
\end{figure}


Fig. \ref{f2} shows RMSE results of $\hat{\beta}_1$ and $\hat{\beta}_2$ from the ECME algorithm, FLOM-MUSIC \cite{THLiu}, and $\ell_p$-MUSIC \cite{Zeng} with the CRLBs in \eqref{22}. Here, $u_1(t)\sim\mathcal{CN}(0,\gamma)$ and $u_2(t)\sim\mathcal{CN}(0,2\gamma)$ are mutually uncorrelated, each RMSE is calculated from $50$ independent trials, and the $\hat{\beta}_g$'s from FLOM-MUSIC and $\ell_p$-MUSIC are obtained by root-MUSIC \cite{Barabell}. We can see that as expected, the RMSE results of $\hat{\beta}_1$ and $\hat{\beta}_2$ from the ECME algorithm asymptotically approach the CRLBs as $\gamma$ (i.e., signal power) increases, which verifies the CRLBs in \eqref{22}. Moreover, each RMSE curve from the ECME algorithm is below those from FLOM-MUSIC and $\ell_p$-MUSIC because ML based DOA estimation outperforms subspace based DOA estimation in terms of accuracy \cite{Krim}.

\begin{figure} 
\centerline{\includegraphics[width=19pc]{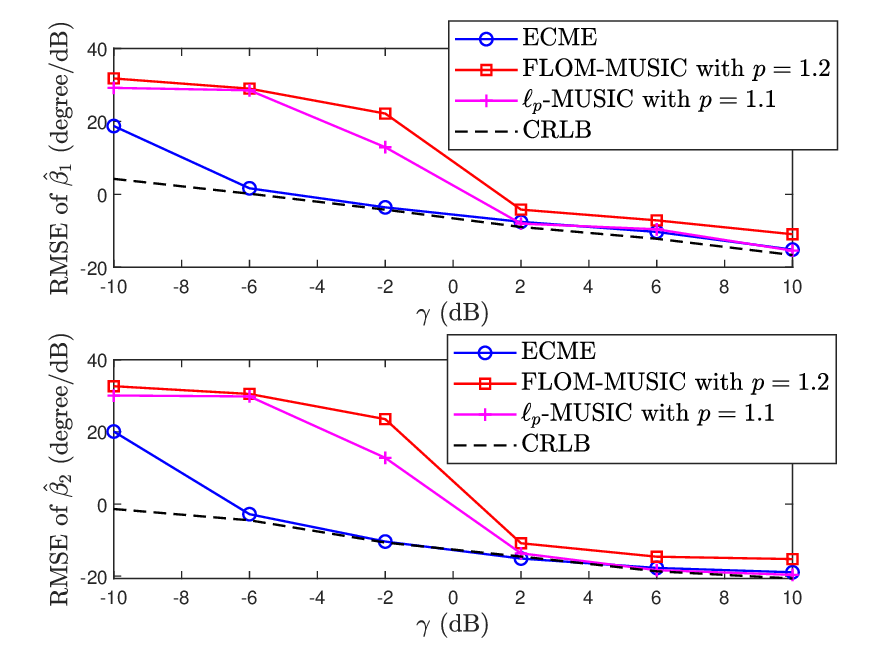}}
\caption{RMSE results of $\hat{\beta}_g$ and $\hat{\beta}_2$ obtained from three methods versus $\gamma$.}\label{f2}
\end{figure}

\section{Conclusion}

The ECME algorithm has been applied and designed to efficiently solve the DMLDF problem in the mixture noise of Gaussian and SI components, which additionally utilizes the actual LLF to first update partial parameter estimates at every iteration and does not need to initialize all parameter estimates. Moreover, the deterministic CRLBs of DOA estimators have been derived. Simulation results have indicated that the ECME algorithm exhibits proper convergence and the CRLBs are correct.

\begin{appendices}
\section{Proof of \textbf{Theorem \ref{t1}}}
\label{fulu}
Because all Fisher information matrix (FIM) cross-terms involving both signal and noise parameters are identically zero, we utilize the actual LLF $f(\boldsymbol{\Phi})$ in \eqref{6} to derive the derivative of $f(\boldsymbol{\Phi})$ with respect to the signal parameter $\beta_g$ by
\begin{eqnarray}  \label{A1}
    &\frac{\partial f}{\partial\beta_g}=2\sum_{t=1}^T\Re\big\{u_g^*(t)\mathbf{z}_g^H\mathbf{Q}_t^{-1}\big[\mathbf{r}(t)-\mathbf{Du}(t)\big]\big\},
\end{eqnarray}
where $g=1,\dots,G$, $u_g^*(t)$ denotes the conjugate of $u_g(t)$, and $\mathbf{z}_g=\frac{d}{d\beta_g}\mathbf{d}(\beta_g)$. Utilizing \eqref{A1}, we can obtain the derivative of $f(\boldsymbol{\Phi})$ with respect to $\boldsymbol{\beta}$ by
\begin{eqnarray} \label{A2}
  &\frac{\partial f}{\partial\boldsymbol{\beta}}=2\sum_{t=1}^T\Re\big\{\mathbf{U}^H(t)\mathbf{Z}^H\mathbf{Q}_t^{-1}\big[\mathbf{r}(t)-\mathbf{Du}(t)\big]\big\}.
\end{eqnarray}
Similarly, we derive the derivative of $f(\boldsymbol{\Phi})$ with respect to the signal parameter $\mathbf{u}^*(i)$ by
\begin{eqnarray}  \label{A3}
    &\frac{\partial f}{\partial\mathbf{u}^*(i)}=\mathbf{D}^H\mathbf{Q}_i^{-1}\big[\mathbf{r}(i)-\mathbf{Du}(i)\big]\nonumber\\
    &=\frac{1}{2}\big(\frac{\partial f}{\partial\Re\{\mathbf{u}(i)\}}+\jmath\frac{\partial f}{\partial\Im\{\mathbf{u}(i)\}}\big),
\end{eqnarray}
which indicates
    \begin{subequations} \label{A4}
    \begin{eqnarray}
        &\frac{\partial f}{\partial\Re\{\mathbf{u}(i)\}}=2\Re\big\{\mathbf{D}^H\mathbf{Q}_i^{-1}\big[\mathbf{r}(i)-\mathbf{Du}(i)\big]\big\},\\
        &\frac{\partial f}{\partial\Im\{\mathbf{u}(i)\}}=2\Im\big\{\mathbf{D}^H\mathbf{Q}_t^{-1}\big[\mathbf{r}(i)-\mathbf{Du}(i)\big]\big\},
    \end{eqnarray}
    \end{subequations}
    where $i=1,\dots,T$. Next, we compute the non-zero submatrices of the FIM matrix unrelated to the noise parameter $(\sigma,\boldsymbol{\tau},\boldsymbol{\Sigma})$ by
    \begin{subequations} \label{A5}
        \begin{eqnarray}
\mathbf{F}_{\boldsymbol{\beta},\boldsymbol{\beta}}&=&\mathcal{E}\Big\{\Big(\frac{\partial f}{\partial\boldsymbol{\beta}}\Big)\Big(\frac{\partial f}{\partial\boldsymbol{\beta}}\Big)^T\Big\}\nonumber\label{A2a}\\
            &=&2{\sum}_{t=1}^T\Re\big\{\mathbf{U}^H(t)\mathbf{Z}^H\mathbf{Q}_t^{-1}\mathbf{Z}\mathbf{U}(t)\big\}\nonumber\\
            &=&\boldsymbol{\Gamma},\\
            \mathbf{F}_{\boldsymbol{\beta},\Re\{\mathbf{u}(i)\}}&=&\mathcal{E}\Big\{\Big(\frac{\partial f}{\partial\boldsymbol{\beta}}\Big)\Big(\frac{\partial f}{\partial\Re\{\mathbf{u}(i)\}}\Big)^T\Big\}\nonumber\\
            &=&2\Re\big\{\mathbf{U}^H(i)\mathbf{Z}^H\mathbf{Q}_i^{-1}\mathbf{D}\big\}\nonumber\\
            &=&\Re\{\boldsymbol{\Lambda}_i\},i=1,\dots,T,\\
            \mathbf{F}_{\boldsymbol{\beta},\Im\{\mathbf{u}(i)\}}&=&\mathcal{E}\Big\{\Big(\frac{\partial f}{\partial\boldsymbol{\beta}}\Big)\Big(\frac{\partial f}{\partial\Im\{\mathbf{u}(i)\}}\Big)^T\Big\}\nonumber\\
            &=&-2\Im\big\{\mathbf{U}^H(i)\mathbf{Z}^H\mathbf{Q}_i^{-1}\mathbf{D}\big\}\nonumber\\
            &=&-\Im\{\boldsymbol{\Lambda}_i\},i=1,\dots,T,\\
            \mathbf{F}_{\Re\{\mathbf{u}(i)\},\Re\{\mathbf{u}(i)\}}&=&\mathcal{E}\Big\{\Big(\frac{\partial f}{\partial\Re\{\mathbf{u}(t)\}}\Big)\Big(\frac{\partial f}{\partial\Re\{\mathbf{u}(i)\}}\Big)^T\Big\}\nonumber\\
            &=&2\Re\big\{\mathbf{D}^H\mathbf{Q}_i^{-1}\mathbf{D}\big\}\nonumber\\
            &=&\Re\{\mathbf{H}_i\},i=1,\dots,T,\\
            \mathbf{F}_{\Re\{\mathbf{u}(i)\},\Im\{\mathbf{u}(i)\}}&=&\mathcal{E}\Big\{\Big(\frac{\partial f}{\partial\Re\{\mathbf{u}(t)\}}\Big)\Big(\frac{\partial f}{\partial\Im\{\mathbf{u}(i)\}}\Big)^T\Big\}\nonumber\\
            &=&-2\Im\big\{\mathbf{D}^H\mathbf{Q}_i^{-1}\mathbf{D}\big\}\nonumber\\
            &=&-\Im\{\mathbf{H}_i\},i=1,\dots,T,\\
            \mathbf{F}_{\Im\{\mathbf{u}(i)\},\Im\{\mathbf{u}(i)\}}&=&\mathcal{E}\Big\{\Big(\frac{\partial f}{\partial\Im\{\mathbf{u}(t)\}}\Big)\Big(\frac{\partial f}{\partial\Im\{\mathbf{u}(i)\}}\Big)^T\Big\}\nonumber\\
            &=&2\Re\big\{\mathbf{D}^H\mathbf{Q}_i^{-1}\mathbf{D}\big\}\nonumber\\
            &=&\Re\{\mathbf{H}_i\},i=1,\dots,T.
        \end{eqnarray}
    \end{subequations}
    Here, we can obtain the FIM matrix by
    \begin{equation}  \label{A6}
        \begin{bmatrix}
            \boldsymbol{\Gamma} & \mathbf{B}_1 & \mathbf{B}_2  & \cdots & \mathbf{B}_T & \mathbf{0} \\
            \mathbf{B}_1^T & \mathbf{E}_1 & \mathbf{0} & \cdots & \mathbf{0} & \mathbf{0}\\
            \mathbf{B}_2^T & \mathbf{0} & \mathbf{E}_2 & \cdots & \mathbf{0} & \mathbf{0} \\
            \vdots & \vdots & \vdots & \ddots & \vdots & \vdots \\
            \mathbf{B}_T^T & \mathbf{0} & \mathbf{0} & \cdots & \mathbf{E}_T & \mathbf{0} \\
            \mathbf{0} & \mathbf{0} & \mathbf{0} & \cdots & \mathbf{0} & \mathbf{F}_{\boldsymbol{\phi},\boldsymbol{\phi}}
            \end{bmatrix}
    \end{equation}
where $\mathbf{F}_{\boldsymbol{\phi},\boldsymbol{\phi}}$ is the submatrix of the FIM matrix only related to the noise parameter $(\sigma,\boldsymbol{\tau},\boldsymbol{\Sigma})$, $\mathbf{B}_i=\big[\Re\{\boldsymbol{\Lambda}_i\}~ -\Im\{\boldsymbol{\Lambda}_i\}\big]$, and
\begin{equation}
    \mathbf{E}_i=
    \begin{bmatrix}
        \Re\{\mathbf{H}_i\} & -\Im\{\mathbf{H}_i\} \\
        -\Im\{\mathbf{H}_i\}^T & \Re\{\mathbf{H}_i\}
    \end{bmatrix}=
    \begin{bmatrix}
        \Re\{\mathbf{H}_i\} & -\Im\{\mathbf{H}_i\} \\
        \Im\{\mathbf{H}_i\} & \Re\{\mathbf{H}_i\}
    \end{bmatrix}.\nonumber
\end{equation}
Applying the partitioned matrix inversion formula, we obtain
\begin{eqnarray} \label{A7}
    &\mathrm{CRLB}^{-1}(\hat{\boldsymbol{\beta}})=\boldsymbol{\Gamma}-\sum_{t=1}^T\mathbf{B}_t\mathbf{E}_t^{-1}\mathbf{B}_t^T.
\end{eqnarray}
Let $\mathbf{G}_i=\mathbf{H}_i^{-1}=(\mathbf{D}^H\mathbf{Q}_i^{-1}\mathbf{D})^{-1}/2$ and 
\begin{equation} \label{A8}
    \mathbf{E}_i^{-1}=
    \begin{bmatrix}
        \Re\{\mathbf{G}_i\} & -\Im\{\mathbf{G}_i\} \\
        \Im\{\mathbf{G}_i\} & \Re\{\mathbf{G}_i\}
    \end{bmatrix},
\end{equation}
which makes
\begin{eqnarray} \label{A9}
    \mathbf{B}_t\mathbf{E}_t^{-1}\mathbf{B}_t^T&=&
    \begin{bmatrix}
        \Re\{\boldsymbol{\Lambda}_t\} & -\Im\{\boldsymbol{\Lambda}_t\}
    \end{bmatrix}\cdot\nonumber\\
    &&\begin{bmatrix}
        \Re\{\mathbf{G}_t\} & -\Im\{\mathbf{G}_t\} \\
        \Im\{\mathbf{G}_t\} & \Re\{\mathbf{G}_t\}
    \end{bmatrix}
    \begin{bmatrix}
        \Re\{\boldsymbol{\Lambda}_t\}^T \\ -\Im\{\boldsymbol{\Lambda}_t\}^T
    \end{bmatrix}\nonumber\\
&=&\begin{bmatrix}
        \Re\{\boldsymbol{\Lambda}_t\} & -\Im\{\boldsymbol{\Lambda}_t\}
    \end{bmatrix}
    \begin{bmatrix}
        \Re\{\mathbf{G}_t\boldsymbol{\Lambda}_t^H\} \\ \Im\{\mathbf{G}_t\boldsymbol{\Lambda}_t^H\}
    \end{bmatrix}\nonumber\\
    &=&\Re\big\{\boldsymbol{\Lambda}_t\mathbf{G}_t\boldsymbol{\Lambda}_t^H\big\}.
\end{eqnarray}
Finally, inserting \eqref{A2a} and \eqref{A9} into \eqref{A7} obtains
\begin{eqnarray}
    \mathrm{CRLB}(\hat{\boldsymbol{\beta}})=\Big(2\sum_{t=1}^T\Re\big\{\mathbf{U}^H(t)\tilde{\mathbf{Z}}_t^H\mathbf{P}_{\tilde{\mathbf{D}}_t}^{\perp}\tilde{\mathbf{Z}}_t\mathbf{U}(t)\big\}\Big)^{-1}.
\end{eqnarray}
The proof is completed.

\end{appendices}

\printbibliography 

\end{document}